\documentclass[10pt,conference]{IEEEtran}

\usepackage{cite}
\usepackage{graphicx}

\usepackage{amsmath}   
\usepackage{amssymb}
\usepackage{amsfonts}

\usepackage{dsfont}
\usepackage[latin1]{inputenc}
\usepackage{times}

\newtheorem{prp}{Proposition}
\newtheorem{thm}{Theorem}
\newtheorem{rem}[prp]{Remark}
\newtheorem{cor}[prp]{Corollary}
\newtheorem{lem}[prp]{Lemma}

\def\mb{\mathbf}
\def\opn{\operatorname*}
\def\bs{\boldsymbol}
\def\bb{\mathbb}
\def\ds{\mathds}
\def\mc{\mathcal}
\def\ss#1{{\sf #1}}

\def\esc#1{\textrm{#1}}

\def\rveci#1#2{\uppercase{#1}_{#2}}
\def\rvecri#1#2{\lowercase{#1}_{#2}}

\def\vec#1{\mb{#1}}
\def\rvec#1{\bs{\uppercase{#1}}} 
\def\rvecr#1{\bs{\lowercase{#1}}} 
\def\vecs#1{\bs{#1}}

\def\mat#1{\mb{\uppercase{#1}}}
\def\mats#1{\bs{#1}}

\def\dim{n}
\def\dimp{m}

\def\R{\ds{R}}

\def\RS{\R} 
\def\Sym#1{\ds{S}^{#1}}
\def\PSD#1{\Sym{#1}_+}
\def\PD#1{\Sym{#1}_{++}}

\def\N#1#2{\mc{N}\left(#1,#2\right)}
\def\NZ#1{\N{\vecs{0}}{#1}}

\def\pdf#1#2{P_{\rvec{#1}}(#2)}
\def\EspOp{\ss{E}}
\def\Esp#1{\EspOp \left\{#1\right\}}  
\def\E_#1#2{\EspOp_{#2}\left\{#1\right\}}  

\def\d{\opn{d}\!}
\def\H{T}
\def\T{T}

\def\nume{\esc{e}}
\def\diag#1{\vec{d}_{#1}}
\def\Diag#1{\mat{D}_{#1}}

\def\I{I} 
\def\Ent{h} 

\def\EP{N} 

\def\preEig{\mats{\Lambda}}
\def\preeig#1{\lambda_{#1}}
\def\preEvec{\vecs{\lambda}}
\def\noiseR{\mats{\Gamma}}
\def\noiser#1{\gamma_{#1}}

\def\CostaT{t}
\def\CosMaT{T}

\def\req#1{(\ref{#1})}
\def\ie{i.e.}
\def\eg{e.g.}

\begin{document}

\title{A multivariate generalization of Costa's entropy power inequality}

\author{\authorblockN{Miquel Payar\'o and Daniel P.~Palomar}
\authorblockA{Department of Electronic
and Computer Engineering,\\ Hong Kong University of Science and
Technology, \\ Clear Water Bay, Kowloon, Hong Kong. \\ \{{\tt
miquel.payaro, palomar}\}{\tt @ust.hk}}}

\maketitle

\begin{abstract}
A simple multivariate version of Costa's entropy power inequality is
proved. In particular, it is shown that if independent white
Gaussian noise is added to an arbitrary multivariate signal, the
entropy power of the resulting random variable is a multidimensional
concave function of the individual variances of the components of
the signal. As a side result, we also give an expression for the
Hessian matrix of the entropy and entropy power functions with
respect to the variances of the signal components, which is an
interesting result in its own right.
\end{abstract}

\section{Introduction}

The entropy power of the random vector $\rvec{Y} \in \R^{\dim}$ was
first introduced by Shannon in his seminal work \cite{shannon:48}
and is, since then, defined as
\begin{gather} \label{eq:EP}
\EP(\rvec{Y}) = \frac{1}{2 \pi \nume} \exp \left(
\frac{2}{\dim} \Ent(\rvec{Y}) \right),
\end{gather}
where $\Ent(\rvec{Y})$ represents the differential entropy, which,
for continuous random vectors reads as\footnote{Throughout this
paper we work with natural logarithms.}
\begin{gather} \nonumber
\Ent(\rvec{Y}) = \Esp{\log \frac{1}{\pdf{\rvec{Y}}{\rvec{Y}}}} 
.
\end{gather}
For the case where the distribution of $\rvec{y}$ assigns positive
mass to one or more singletons in $\R^{\dim}$, the above definition
is extended with $\Ent(\rvec{Y}) = -\infty$.

The entropy power of a random vector $\rvec{Y}$ represents the
variance (or power) of a standard Gaussian random vector $\rvec{Y}_G
\sim \NZ{\sigma^2\mat{I}_\dim}$ such that both $\rvec{Y}$ and
$\rvec{Y}_G$ have identical differential entropy, $\Ent(\rvec{Y}_G)
= \Ent(\rvec{Y})$.

\subsection{Shannon's entropy power inequality (EPI)}

For any two independent arbitrary random vectors $\rvec{x} \in
\R^\dim$ and $\rvec{w} \in \R^\dim$, Shannon gave in
\cite{shannon:48} the following inequality:
\begin{gather} \nonumber 
\EP(\rvec{x} + \rvec{w}) \geq \EP(\rvec{x}) + \EP(\rvec{w}).
\end{gather}
The first rigorous proof of Shannon's EPI was given in
\cite{stam:59} by Stam, and was simplified by Blachman in
\cite{blachman:65}. A simple and very elegant proof by Verdú and Guo
based on estimation theoretic considerations has recently appeared
in \cite{verdu:06}.

Among many other important results, Bergmans' proof of the converse
for the degraded Gaussian broadcast channel \cite{bergmans:74} and
Oohama's partial solution to the rate distortion region problem for
Gaussian multiterminal source coding systems \cite{oohama:05} follow
from Shannon's EPI.

\subsection{Costa's EPI}

Under the setting of Shannon's EPI, Costa proved in \cite{costa:85}
that, provided that the random vector $\rvec{w}$ is white Gaussian
distributed, then Shannon's EPI can be strengthened to
\begin{gather} \label{eq:CostaEPIO}
\EP(\rvec{x} + \sqrt{\CostaT}\rvec{w}) \geq (1-\CostaT)\EP(\rvec{x}) + \CostaT\EP(\rvec{x} + \rvec{w}),
\end{gather}
where $t\in[0,1]$. As Costa noted, the above EPI is equivalent to
the concavity of the entropy power function $\EP(\rvec{x} +
\sqrt{\CostaT} \rvec{w})$ with respect to the parameter $\CostaT$,
or, formally\footnote{The equivalence between equations
\req{eq:CostaEPIO} and \req{eq:CostaEPI} is due to the fact that the
function $\EP\big(\rvec{x} + \sqrt{\CostaT} \rvec{w}\big)$ is twice
differentiable almost everywhere thanks to the smoothing properties
of the added Gaussian noise.}
\begin{gather} \label{eq:CostaEPI}
\frac{\d^2}{\d \CostaT^2} \EP\big(\rvec{x} + \sqrt{\CostaT}
\rvec{w}\big) \leq 0.
\end{gather}
Due to its inherent interest and to the fact that the proof by Costa
was rather involved, simplified proofs of his result have been
subsequently given in \cite{dembo:89, villani:00, guo:06, rioul:07}.

Additionally, in his paper Costa presented two extensions of his
main result in \req{eq:CostaEPI}. Precisely, he showed that the EPI
is also valid when the Gaussian vector $\rvec{w}$ is not white, and
also for the case where the $\CostaT$ parameter is multiplying the
arbitrarily distributed random vector $\rvec{X}$,
\begin{gather} \label{eq:tXpW}
\frac{\d^2}{\d \CostaT^2} \EP\big(\sqrt{\CostaT} \rvec{x} +
\rvec{w}\big) \leq 0.
\end{gather}

Similarly to Shannon's EPI, Costa's EPI has been used successfully
to derive important information-theoretic results concerning, \eg,
Gaussian interference channels in \cite{costa:85b} or multi-antenna
flat fading channels with memory in \cite{lapidoth:03}.

\subsection{Aim of the paper}

Our objective is to extend the particular case in \req{eq:tXpW} of
Costa's EPI to the multivariate case, allowing the real parameter
$\CostaT\in\R$ to become a matrix
$\mat{\CosMaT}\in\R^{\dim\times\dim}$, which, to the best of the
authors' knowledge, has not been considered before.

Beyond its theoretical interest, the motivation behind this study is
due to the fact that the concavity of the entropy power with respect
to $\mat{\CosMaT}$ implies the concavity of the entropy and mutual
information quantities, which would be a very desirable property in
optimization procedures in order to be able to, \eg, design the
linear precoder that maximizes the mutual information in the linear
vector Gaussian channel with arbitrary input distributions.

Consequently, we investigate the concavity of the function
\begin{gather} \label{eq:general_case_X}
\EP \big(\mat{\CosMaT}^{1/2}\rvec{x} + \rvec{w}\big),
\end{gather}
with respect to the symmetric matrix $\mat{\CosMaT} =
\mat{\CosMaT}^{1/2}\mat{\CosMaT}^{\T/2}$. Unfortunately, the concavity in $\mat{\CosMaT}$ of the entropy power can be easily disproved by finding simple counterexamples as in \cite{payaro:08} or even through numerical computations of the entropy power. 

Knowing this negative result, we thus focus our study on the next
possible multivariate candidate: a diagonal matrix. Our objective
now is to study the concavity of
\begin{gather} \label{eq:EPX}
\EP \big( \preEig^{1/2}\rvec{X} + \rvec{W} \big),
\end{gather}
w.r.t.~the diagonal matrix $\preEig = \opn{diag}(\preEvec)$, with
$[\preEvec]_i = \preeig{i}$.

For the sake of notation, throughout this work we
define
\begin{gather} \nonumber 
\rvec{Y} = \preEig^{1/2}\rvec{X} + \rvec{W},
\end{gather}
where we recall that the random vector $\rvec{W}$ is assumed to
follow a white zero-mean Gaussian distribution and the distribution
of the random vector $\rvec{X}$ is arbitrary. In particular, the
distribution of $\rvec{X}$ is allowed to assign positive mass to one
or more singletons in $\R^{\dim}$. Consequently, the results
presented in Theorems \ref{thm:MV_EI} and \ref{thm:MV_EPI} in
Section \ref{sec:results} also hold for the case where the random
vector $\rvec{X}$ is discrete.

\section{Mathematical preliminaries}

In this section we present a number of lemmas followed by a
proposition that will prove useful in the proof of our
multidimensional EPI. In our derivations, the identity matrix is
denoted by $\mat{I}$, the vector with all its entries equal to $1$
is represented by $\vecs{1}$, and $\mat{A}\circ\mat{B}$ represents
the Hadamard (or Schur) element-wise matrix product.

\begin{lem}[Bhatia {\cite[p.~15]{bhatia:07}}] \label{lem:XXXX}
Let $\mat{A} \in \PSD{\dim}$ be a positive
semidefinite matrix, $\mat{A} \geq \mats{0}$. Then it follows that
\begin{gather} \nonumber
\left[
\begin{array}{cc}
\mat{A} & \mat{A} \\
\mat{A} & \mat{A}
\end{array}
\right] \geq \mats{0}.
\end{gather}
\end{lem}
\vspace{.1cm}
\begin{proof}
Since $\mat{A} \geq \mats{0}$, consider $\mat{A} =
\mat{C}\mat{C}^\T$ and write
\begin{gather} \nonumber
\left[
\begin{array}{cc}
\mat{A} & \mat{A} \\
\mat{A} & \mat{A}
\end{array}
\right] = \left[
\begin{array}{c}
\mat{C} \\
\mat{C}
\end{array}
\right]\left[
\begin{array}{cc}
\mat{C}^\T & \mat{C}^\T
\end{array}
\right].
\end{gather}
\end{proof}

\begin{lem}[Bhatia {\cite[Exercise 1.3.10]{bhatia:07}}] \label{lem:XIIinvX}
Let $\mat{A} \in \PD{\dim}$ be a positive definite matrix, $\mat{A}
> \mats{0}$. Then it follows that
\begin{gather} \label{eq:XIIinvX}
\left[
\begin{array}{cc}
\mat{A} & \mat{I} \\
\mat{I} & \mat{A}^{-1}
\end{array}
\right] \geq \mats{0}.
\end{gather}
\end{lem}
\vspace{.1cm}
\begin{proof}
Consider again $\mat{A} = \mat{C}\mat{C}^\T$, then we have
$\mat{A}^{-1} = \mat{C}^{-\T}\mat{C}^{-1}$. Now, simply write
\req{eq:XIIinvX} as
\begin{gather} \nonumber
\left[
\begin{array}{cc}
\mat{A} & \mat{I} \\
\mat{I} & \mat{A}^{-1}
\end{array}
\right] 
= \left[
\begin{array}{cc}
\mat{C} & \mats{0} \\
\mats{0} & \mat{C}^{-\T}
\end{array}
\right]\left[
\begin{array}{cc}
\mat{I} & \mat{I} \\
\mat{I} & \mat{I}
\end{array}
\right]\left[
\begin{array}{cc}
\mat{C}^{\T} & \mats{0} \\
\mats{0} & \mat{C}^{-1}
\end{array}
\right], \nonumber
\end{gather}
which, from Sylvester's law of inertia for congruent matrices
\cite[p.~5]{bhatia:07} and Lemma \ref{lem:XXXX}, is positive
semidefinite.
\end{proof}

\begin{lem}[Schur Theorem] \label{lem:schur_th}
If the matrices $\mat{A}$ and $\mat{B}$ are positive semidefinite,
then so is the product $\mat{A}\circ\mat{B}$. If, both $\mat{A}$ and
$\mat{B}$ are positive definite, then so is $\mat{A}\circ\mat{B}$.
In other words, the class of positive (semi)definite matrices is
closed under the Hadamard product.
\end{lem}
\begin{proof}
See \cite[Th. 7.5.3]{horn:85} or \cite[Th. 5.2.1]{horn:91}.
\end{proof}

\begin{lem}[Schur complement] \label{lem:schur_compl}
Let the matrices $\mat{A}\in \PD{\dim}$ and $\mat{B}\in \PD{\dimp}$
be positive definite, $\mat{A}>\mats{0}$ and $\mat{B}>\mats{0}$, and
not necessarily of the same dimension. Then the following statements
are equivalent
\begin{enumerate}
\item
$\left[
\begin{array}{cc}
\mat{A} & \mat{D} \\
\mat{D}^\H & \mat{B}
\end{array}
\right] \geq 0,$

\item $\mat{B} \geq \mat{D}^\H \mat{A}^{-1} \mat{D}$,

\item $\mat{A} \geq \mat{D} \mat{B}^{-1} \mat{D}^\H$,
\end{enumerate}
where $\mat{D}\in\R^{\dim\times\dimp}$ is any arbitrary matrix.
\end{lem}
\begin{proof}
See \cite[Th.~7.7.6]{horn:85} and the second exercise following it
or \cite[Prop.~8.2.3]{bernstein:05}.
\end{proof}

With the above lemmas at hand, we are now ready to prove the
following proposition:
\begin{prp} \label{prp:XoinvZ}
Consider two positive definite matrices $\mat{A} \in \PD{\dim}$ and
$\mat{B} \in \PD{\dim}$ of the same dimension, and let
$\Diag{\mat{A}}$ be a diagonal matrix containing the diagonal
elements of $\mat{A}$, (\ie, $\Diag{\mat{A}} =
\mat{A}\circ\mat{I}$). Then it follows that
\begin{gather} \label{eq:AoinvB}
\mat{A}\circ\mat{B}^{-1} \geq \Diag{\mat{A}}
\left(\mat{A}\circ\mat{B}\right)^{-1} \Diag{\mat{A}}.
\end{gather}
\end{prp}
\vspace{.1cm}
\begin{proof}
From Lemmas \ref{lem:XXXX}, \ref{lem:XIIinvX}, and
\ref{lem:schur_th}, it follows that
\begin{gather} \nonumber
\left[
\begin{array}{cc}
\mat{A} & \mat{A} \\
\mat{A} & \mat{A}
\end{array}
\right] \circ \left[
\begin{array}{cc}
\mat{B} & \mat{I} \\
\mat{I} & \mat{B}^{-1}
\end{array}
\right] = \left[
\begin{array}{cc}
\mat{A}\circ\mat{B} & \Diag{\mat{A}} \\
\Diag{\mat{A}} & \mat{A}\circ\mat{B}^{-1}
\end{array}
\right] \geq 0.
\end{gather}
Now, from Lemma \ref{lem:schur_compl}, the result follows directly.
\end{proof}

\begin{cor} \label{cor:dinvXXd}
Let $\mat{A} \in \PD{\dim}$ be a positive definite matrix. Then,
\begin{gather}
\diag{\mat{A}}^\H \left(\mat{A}\circ\mat{A}\right)^{-1}
\diag{\mat{A}} \leq \dim, \label{eq:dinvXXd}
\end{gather}
where we have defined $\diag{\mat{A}} = \Diag{\mat{A}}\mats{1} =
(\mat{A}\circ\mat{I})\mats{1}$ as a column vector with the diagonal
elements of matrix $\mat{A}$.
\end{cor}
\begin{proof}
Particularizing the result in Proposition \ref{prp:XoinvZ} with
$\mat{B} = \mat{A}$ and pre- and post-multiplying it by
$\mats{1}^\T$ and $\mats{1}$ we obtain
\begin{gather} \nonumber
\mats{1}^\T\left(\mat{A}\circ\mat{A}^{-1}\right)\mats{1} \geq
\mats{1}^\T \Diag{\mat{A}} \left(\mat{A}\circ\mat{A}\right)^{-1}
\Diag{\mat{A}} \mats{1}.
\end{gather}
The result in \req{eq:dinvXXd} now follows straightforwardly from
the fact $\vecs{1}^\H \left( \mat{A}\circ\mat{A}^{-\T}
\right)\vecs{1} = \dim$, \cite{johnson:86} (see also \cite[Fact
7.6.10]{bernstein:05}, \cite[Lemma 5.4.2(a)]{horn:91}). Note that
$\mat{A}$ is symmetric and thus $\mat{A}^{\T} = \mat{A}$ and
$\mat{A}^{-\T} = \mat{A}^{-1}$.
\end{proof}
\begin{rem}
Note that the proof of Corollary \ref{cor:dinvXXd} is based on the
result of Proposition \ref{prp:XoinvZ} in \req{eq:AoinvB}. An
alternative proof could follow similarly from a different inequality
by Styan in \cite{styan:73}
\begin{gather} \nonumber 
\mat{R}\circ\mat{R}^{-1} + \mat{I} \geq
2\left(\mat{R}\circ\mat{R}\right)^{-1},
\end{gather}
where $\mat{R}$ is constrained to be a correlation matrix
$\mat{R}\circ\mat{I} = \mat{I}$.
\end{rem}

\begin{prp} \label{prp:XXones}
Consider now the positive semidefinite matrix $\mat{A} \in
\PSD{\dim}$. Then,
\begin{gather} \nonumber
\mat{A}\circ\mat{A} \geq
\frac{\diag{\mat{A}}\diag{\mat{A}}^\H}{\dim}.
\end{gather}
\end{prp}
\vspace{.1cm}
\begin{proof}
For the case where $\mat{A} \in \PD{\dim}$ is positive definite,
from \req{eq:dinvXXd} in Corollary \ref{cor:dinvXXd} and Lemma
\ref{lem:schur_compl}, it follows that
\begin{gather} \nonumber
\left[
\begin{array}{cc}
\mat{A}\circ\mat{A} & \diag{\mat{A}} \\
\diag{\mat{A}}^\H & \dim
\end{array}
\right] \geq 0.
\end{gather}
Applying again Lemma \ref{lem:schur_compl}, we get
\begin{gather} \label{eq:XXones_proof}
\mat{A}\circ\mat{A} \geq
\frac{\diag{\mat{A}}\diag{\mat{A}}^\H}{\dim}.
\end{gather}
Now, assume that $\mat{A} \in \PSD{\dim}$ is positive semidefinite.
We thus define $\epsilon > 0$ and consider the positive definite
matrix $\mat{A} + \epsilon\mat{I}$. From \req{eq:XXones_proof}, we
know that
\begin{gather} \nonumber
\left(\mat{A}+ \epsilon\mat{I}\right) \circ \left( \mat{A} +
\epsilon\mat{I} \right) \geq \frac{\diag{\mat{A}+
\epsilon\mat{I}}\diag{\mat{A}+ \epsilon\mat{I}}^\H}{\dim}.
\end{gather}
Taking the limit as $\epsilon$ tends to 0, from continuity, the validity of \req{eq:XXones_proof} for positive
semidefinite matrices follows.
\end{proof}

Finally, to end this section about mathematical prelimi\-na\-ries,
we give a very brief overview on some basic definitions related to
minimum mean-square error (MMSE) estimation. These definitions are
useful in our further derivations due to the relation between the
entropy and the MMSE unveiled in \cite{guo:05}.\footnote{Strictly
speaking the relation found in \cite{guo:05} concerns the quantities
of mutual information and MMSE, but it is still useful for our
problem because the entropy $\Ent(\rvec{y})$ and the mutual
information $\I(\rvec{x}; \rvec{y})$ have the same dependence on
$\preEig$ up to a constant additive term.} Next, we give a lemma
concerning the positive semidefiniteness of a certain class of
matrices closely related with MMSE estimation.

Consider the setting described in the introduction, $\rvec{Y} =
\preEig^{1/2}\rvec{X} + \rvec{W}$. For a given realization of the
observations vector $\rvec{Y} = \rvecr{y}$, the MMSE estimator,
$\widehat{\rvec{X}}(\rvecr{y})$, is given by the conditional mean 
\begin{gather} \nonumber
\widehat{\rvec{X}}(\rvecr{y}) = \Esp{\rvec{x}|\rvec{y} = \rvecr{y}}.
\end{gather}
We now define the conditional MMSE matrix, $\mats{\Phi}_{\rvec{X}}(\rvecr{y})$, as the mean-square error matrix conditioned on the fact that the received vector is equal to $\rvec{Y} = \rvecr{y}$. Formally
\begin{align}
\mats{\Phi}_{\rvec{X}}(\rvecr{y}) &\triangleq \Esp{(\rvec{X} -
\widehat{\rvec{X}}(\rvecr{y})) (\rvec{X} -
\widehat{\rvec{X}}(\rvecr{y}))^\T\big|\rvec{Y}= \rvecr{y}} \nonumber
\\ & = \Esp{\rvec{X}\rvec{X}\phantom{}^\T \big| \rvec{Y}= \rvecr{y}}
\label{eq:mmsecy} \\ & \quad \quad \quad - \Esp{\rvec{X}|\rvec{Y}=
\rvecr{y}} \Esp{\rvec{X}\phantom{}^\T\big|\rvec{Y}= \rvecr{y}}.
\nonumber
\end{align}
From this definition, it is clear that
$\mats{\Phi}_{\rvec{X}}(\rvecr{y})$ is a positive semi-definite
matrix.

Now, the MMSE matrix $\mat{E}_{\rvec{X}}$ can be calculated by
averaging $\mats{\Phi}_{\rvec{X}}(\rvecr{y})$ in \req{eq:mmsecy}
with respect to the distribution of vector $\rvec{y}$ as
\begin{gather} \label{eq:mmse}
\mat{E}_{\rvec{X}} = \Esp{\mats{\Phi}_{\rvec{X}}(\rvec{y})}.
\end{gather}
See below the last lemma in this section.
\begin{lem} \label{lem:psd}
For a given random vector $\rvec{x} \in \RS^{\dim}$, it follows that
$\Esp{\rvec{X}\rvec{X}\phantom{}^\T} \geq
\Esp{\rvec{X}}\Esp{\rvec{X}\phantom{}^\T}$.
\end{lem}
\begin{proof}
Simply note that
\begin{multline} \nonumber
\Esp{\rvec{X}\rvec{X}\phantom{}^\T} - \Esp{\rvec{X}}\Esp{\rvec{X}\phantom{}^\T} \\
= \Esp{(\rvec{X} - \Esp{\rvec{X}})(\rvec{X} - \Esp{\rvec{X}})^\T}
\geq \mats{0}, \nonumber
\end{multline}
where last inequality follows from the fact that the expectation
operator preserves positive semidefiniteness.
\end{proof}

\section{Main result of the paper}
\label{sec:results}

Once all the mathematical preliminaries have been presented, in this
section we give the main result of the paper, namely, the concavity
of the entropy power function $\EP(\rvec{y})$ in \req{eq:EPX}, with
respect to the diagonal elements of $\preEig$. Prior to proving this
result, we present a weaker result concerning the concavity of the
entropy function $\Ent(\rvec{Y})$, which is key in proving the
concavity of the entropy power.

\subsection{Warm up: An entropy inequality}

\begin{thm} \label{thm:MV_EI}
Assume $\rvec{Y} = \preEig^{1/2}\rvec{X} + \rvec{W}$, where
$\rvec{X}$ is arbitrarily distributed and $\rvec{W}$ follows a
zero-mean white Gaussian distribution. Then the entropy
$\Ent(\rvec{Y})$ is a concave function of the diagonal elements of
$\preEig$, \ie,
\begin{gather} \nonumber
\nabla_{\preEvec}^2 \Ent(\rvec{Y}) \leq \mats{0}.
\end{gather}
Furthermore, the entries of the Hessian matrix of the entropy
function $\Ent(\rvec{y})$ with res\-pect to $\preEvec$ are given by
\begin{multline} \label{eq:HessEnt_ij}
\left[\nabla_{\preEvec}^2 \Ent(\rvec{Y})\right]_{ij} \\ =
-\frac{1}{2} \Esp{ \big(\Esp{\rveci{x}{i}\rveci{x}{j} | \rvec{Y}} -
\Esp{\rveci{x}{i}|\rvec{Y}} \Esp{\rveci{x}{j}|\rvec{Y}}\big)^2},
\end{multline}
which can be written more compactly as
\begin{gather} \label{eq:HessEnt}
\nabla_{\preEvec}^2 \Ent(\rvec{Y}) =
-\frac{1}{2} \Esp{\mats{\Phi}_{\rvec{X}}(\rvec{y}) \circ \mats{\Phi}_{\rvec{X}}(\rvec{y})}.
\end{gather}
\end{thm}
\vspace{.1cm}
\begin{proof}
For the computations leading to \req{eq:HessEnt_ij} and
\req{eq:HessEnt} see Appendix \ref{sec:ddent}. Once the expression
in \req{eq:HessEnt} is obtained, concavity (or negative
semidefiniteness of the Hessian matrix) follows straightforwardly
taking into account that the matrix $\mats{\Phi}_{\rvec{X}}(\rvecr{y})$ defined in \req{eq:mmsecy} is positive semidefinite
$\forall \rvecr{y}$, Lemma \ref{lem:schur_th}, and from the fact
that the expectation operator preserves the semidefiniteness.
\end{proof}

\subsection{Multivariate extension of Costa's EPI}
\begin{thm} \label{thm:MV_EPI}
Assume $\rvec{Y} = \preEig^{1/2}\rvec{X} + \rvec{W}$, where
$\rvec{X}$ is arbitrarily distributed and $\rvec{W}$ follows a
zero-mean white Gaussian distribution. Then the entropy power
$\EP(\rvec{Y})$ is a concave function of the diagonal elements of
$\preEig$, \ie,
\begin{gather} \nonumber
\nabla_{\preEvec}^2 \EP(\rvec{Y}) \leq \mats{0}.
\end{gather}
Moreover, the Hessian matrix of the entropy power function
$\EP(\rvec{y})$ with respect to $\preEvec$ is given by
\begin{multline} \label{eq:N}
\nabla_{\preEvec}^2 \EP(\rvec{Y}) \\ = \frac{\EP(\rvec{Y})}{\dim}
\left( \frac{\diag{\mat{E}_{\rvec{X}}}
\diag{\mat{E}_{\rvec{X}}}^\T}{\dim} - \Esp{\mats{\Phi}_{\rvec{X}}(\rvec{y})
\circ \mats{\Phi}_{\rvec{X}}(\rvec{y})}\right),
\end{multline}
where we recall that $\diag{\mat{E}_{\rvec{X}}}$ is a column vector
with the diagonal entries of the matrix $\mat{E}_{\rvec{X}}$ defined in \req{eq:mmse}.
\end{thm}
\begin{proof}
First, let us prove \req{eq:N}. From the definition of the entropy
power in \req{eq:EP} and applying the chain rule we obtain
\begin{gather} \nonumber
\nabla_{\preEvec}^2 \EP(\rvec{Y}) = \frac{2\EP(\rvec{Y})}{\dim}
\left( \frac{2\nabla_{\preEvec}\Ent(\rvec{Y}) \nabla_{\preEvec}^\T
\Ent(\rvec{Y})}{\dim} + \nabla_{\preEvec}^2 \Ent(\rvec{Y}) \right).
\end{gather}
Now, replacing $\nabla_{\preEvec}\Ent(\rvec{Y})$ by its expression
from \cite[Eq. (61)]{guo:05}
\begin{align} \nonumber
[\nabla_{\preEvec}\Ent(\rvec{y})]_i &= \frac{1}{2}
[\mat{E}_{\rvec{x}}]_{ii} = \frac{1}{2}
\Esp{[\mats{\Phi}_{\rvec{X}}(\rvec{y})]_{ii}}, 
\end{align}
and incorporating the expression for
$\nabla_{\preEvec}^2\Ent(\rvec{Y})$ calculated in \req{eq:HessEnt},
the result in \req{eq:N} follows.

Now that a explicit expression for the Hessian matrix has been
obtained, we wish to prove that it is negative semidefinite. Note
from \req{eq:N} that, except for a positive factor, the Hessian
matrix $\nabla_{\preEvec}^2 \EP(\rvec{Y})$ is the sum of a rank one
positive semidefinite matrix and the Hessian matrix of the entropy,
which is negative semidefinite according to Theorem \ref{thm:MV_EI}.
Consequently, the definiteness of $\nabla_{\preEvec}^2
\EP(\rvec{Y})$ is unknown a priori, and some further developments
are needed to determine it, which is what we do next.

Consider a family of positive semidefinite
matrices $\mat{A} \in \PSD{\dim}$, characterized by a
certain vector parameter $\vec{v}$, $\mat{A} = \mat{A}(\vec{v})$.
Applying Pro\-po\-si\-tion \ref{prp:XXones} to each matrix in this family,
we obtain
\begin{gather} \label{eq:Av}
\mat{A}(\vec{v}) \circ \mat{A}(\vec{v}) \geq
\frac{\diag{\mat{A}(\vec{v})}\diag{\mat{A}(\vec{v})}^\T}{\dim}.
\end{gather}
Since \req{eq:Av} is true for all possible values of $\vec{v}$, we
have
\begin{gather} \label{eq:EspAv}
\Esp{\mat{A}(\rvec{v}) \circ \mat{A}(\rvec{v})} \geq
\frac{\Esp{\diag{\mat{A}(\rvec{v})}\diag{\mat{A}(\rvec{v})}^\T}}{\dim},
\end{gather}
where now the parameter $\vec{v}$ has been considered to be a random
variable, $\rvec{v}$. Note that the distribution of $\rvec{v}$ is
arbitrary and does not affect the validity of \req{eq:EspAv}. From
Lemma \ref{lem:psd} we know that
\begin{gather} \nonumber
\Esp{\diag{\mat{A}(\rvec{v})}\diag{\mat{A}(\rvec{v})}^\T} \geq
\Esp{\diag{\mat{A}(\rvec{v})}}\Esp{\diag{\mat{A}(\rvec{v})}^\T},
\end{gather}
from which it follows that
\begin{gather} \nonumber
\Esp{\mat{A}(\rvec{v}) \circ \mat{A}(\rvec{v})} \geq
\frac{\Esp{\diag{\mat{A}(\rvec{v})}}\Esp{\diag{\mat{A}(\rvec{v})}^\T}}{\dim}.
\end{gather}
Since the operators $\diag{\mat{A}}$ and expectation commute we
finally obtain
\begin{gather} \nonumber
\Esp{\mat{A}(\rvec{v}) \circ \mat{A}(\rvec{v})} \geq
\frac{\diag{\Esp{\mat{A}(\rvec{v})}}\diag{\Esp{\mat{A}(\rvec{v})}}^\T}{\dim}.
\end{gather}

Identifying $\mat{A}(\rvec{v})$ with the random covariance error
matrix $\mats{\Phi}_{\rvec{X}}(\rvec{Y})$ and using \req{eq:mmse} the
result in the theorem follows as
\begin{gather} \nonumber
\frac{\diag{\mat{E}_{\rvec{X}}} \diag{\mat{E}_{\rvec{X}}}^\T}{\dim} -
\Esp{\mats{\Phi}_{\rvec{X}}(\rvec{Y}) \circ \mats{\Phi}_{\rvec{X}}(\rvec{Y})}
\leq \mats{0},
\end{gather}
and $\EP(\rvec{Y}) \geq 0$.
\end{proof}

\section{Conclusion}
In this paper we have proved that, for $\rvec{Y} =
\preEig^{1/2}\rvec{X} + \rvec{W}$  the functions $\EP(\rvec{y})$ and
$\Ent(\rvec{y})$ are concave with respect to the diagonal entries of
$\preEig$ and have also given explicit expressions for the elements
of the Hessian matrices $\nabla_{\preEvec}^2 \EP(\rvec{Y})$ and
$\nabla_{\preEvec}^2 \Ent(\rvec{Y})$.

Besides its theoretical interest and inherent beauty, the importance
of the results presented in this work lie mainly in their potential
applications, such as, the calculation of the optimal power
allocation to maximize the mutual information for a given
non-Gaussian constellation as described in \cite{payaro:08}.


\appendices

\section{Calculation of $\nabla_{\preEvec}^2 \Ent(\rvec{Y})$} \label{sec:ddent}

In this section we are interested in the calculation of the elements
of the Hessian matrix $\left[\nabla_{\preEvec}^2
\Ent(\rvec{Y})\right]_{ij}$, which are defined by
\begin{gather} \nonumber
\left[\nabla_{\preEvec}^2 \Ent(\rvec{Y})\right]_{ij} =
\frac{\partial^2 \Ent(\preEig^{1/2}\rvec{X} + \rvec{W})}{\partial
\preeig{i} \partial \preeig{j}}.
\end{gather}

First of all, using the properties of differential entropy we write
\begin{gather} \nonumber
\Ent(\preEig^{1/2}\rvec{X} + \rvec{W}) = \Ent(\rvec{X} +
\preEig^{-1/2} \rvec{W}) + \frac{1}{2}\log|\preEig|,
\end{gather}
and recalling that we work with natural logarithms we have
\begin{gather} \label{eq:ddhP}
\frac{\partial^2 \Ent(\preEig^{1/2}\rvec{X} + \rvec{W})}{\partial
\preeig{i} \partial \preeig{j}} = \frac{\partial^2 \Ent(\rvec{X} +
\preEig^{-1/2}\rvec{W})}{\partial \preeig{i}
\partial \preeig{j}} - \frac{\delta_{ij}}{2\preeig{i}^2}.
\end{gather}
We are now interested in expanding the first term in the right hand
side of last equation, so we define the diagonal matrix $\noiseR =
\preEig^{-1}$ and the random vector $\rvec{z} =
\preEig^{-1/2}\rvec{y}$. Thus $[\noiseR]_{ii} = \noiser{i} =
1/\preeig{i}$ and $\rvec{Z} = \rvec{X} + \preEig^{-1/2}\rvec{W} =
\rvec{X} + \noiseR^{1/2}\rvec{W}$.

Applying the chain rule we obtain
\begin{multline} \label{eq:chain_rule}
\frac{\partial^2 \Ent(\rvec{X} + \preEig^{-1/2}\rvec{W})}{\partial
\preeig{i} \partial \preeig{j}} =  \frac{1}{\preeig{i}^2
\preeig{j}^2} \frac{\partial^2 \Ent(\rvec{X} +
\noiseR^{1/2}\rvec{W})}{\partial \noiser{i} \partial
\noiser{j}}\Big|_{\noiseR = \preEig^{-1}} \\ +
\frac{2\delta_{ij}}{\preeig{i}^3}\frac{\partial \Ent(\rvec{X} +
\noiseR^{1/2}\rvec{W})}{\partial \noiser{i}}\Big|_{\noiseR =
\preEig^{-1}}.
\end{multline}
The expressions for the two terms
\begin{gather} \nonumber
\frac{\partial^2 \Ent(\rvec{X} + \noiseR^{1/2}\rvec{W})}{\partial
\noiser{i} \partial \noiser{j}} \quad \textrm{and} \quad
\frac{\partial \Ent(\rvec{X} + \noiseR^{1/2}\rvec{W})}{\partial
\noiser{i}}
\end{gather}
are given in Appendix \ref{sec:Bakry}, where we also sketch how they
can be computed, for further details see \cite{payaro:08}. Using
these results, the right hand side of the expression in
\req{eq:chain_rule} can be rewritten as
\begin{multline} \nonumber
\frac{1}{\preeig{i}^2 \preeig{j}^2} \left( -\frac{1}{2} \Esp{ \frac{
\left(\Esp{\rveci{x}{i}\rveci{x}{j} | \rvec{Z}} -
\Esp{\rveci{x}{i}|\rvec{Z}}
\Esp{\rveci{x}{j}|\rvec{Z}}\right)^2}{\noiser{i}^2\noiser{j}^2}}
\right.
\\ \left. - \frac{\delta_{ij}}{2\noiser{i}^2} + \Esp{\frac{ \Esp{\rveci{x}{i}^2
|\rvec{Z}} - (\Esp{\rveci{x}{i}|\rvec{Z}})^2
}{\noiser{i}^3}}\delta_{ij} \right)\Bigg|_{\noiseR = \preEig^{-1}}
\\ + \frac{2\delta_{ij}}{\preeig{i}^3}\left(
\frac{1}{2\noiser{i}^2} \left(\noiser{i} - \Esp{(\rveci{x}{i} -
\Esp{\rveci{x}{i}|\rvec{Z}})^2} \right) \right) \Bigg|_{\noiseR =
\preEig^{-1}}.
\end{multline}
Simplifying terms we obtain
\begin{multline} \label{eq:ddhN}
-\frac{1}{2} \Esp{ \left(\Esp{\rveci{x}{i}\rveci{x}{j} | \rvec{Z}} -
\Esp{\rveci{x}{i}|\rvec{Z}} \Esp{\rveci{x}{j}|\rvec{Z}}\right)^2}
\\ - \frac{\delta_{ij}}{2\preeig{i}^2} + \frac{\delta_{ij}}{\preeig{i}} \Esp{\Esp{\rveci{x}{i}^2
|\rvec{Z}} - (\Esp{\rveci{x}{i}|\rvec{Z}})^2}
\\ + \frac{\delta_{ij}}{\preeig{i}^2} - \frac{\delta_{ij}}{\preeig{i}}\Esp{(\rveci{x}{i} -
\Esp{\rveci{x}{i}|\rvec{Z}})^2}.
\end{multline}

Finally, noting that
\begin{gather} \nonumber
\Esp{(\rveci{x}{i} - \Esp{\rveci{x}{i}|\rvec{Z}})^2} =
\Esp{\Esp{\rveci{x}{i}^2 |\rvec{Z}} - (\Esp{\rveci{x}{i}|\rvec{Z}})^2} \\
\Esp{f(\rvec{X})|\rvec{Z}} =
\bb{E}\big\{f(\rvec{X})|\preEig^{1/2}\rvec{Z}\big\} =
\Esp{f(\rvec{X})|\rvec{Y}}, \nonumber
\end{gather}
and plugging \req{eq:ddhN} in \req{eq:ddhP} we obtain the desired
result in \req{eq:HessEnt_ij}:
\begin{multline} \nonumber
\frac{\partial^2 \Ent(\preEig^{1/2}\rvec{X} + \rvec{W})}{\partial
\preeig{i} \partial \preeig{j}} \\ = -\frac{1}{2} \Esp{
\left(\Esp{\rveci{x}{i}\rveci{x}{j} | \rvec{Y}} -
\Esp{\rveci{x}{i}|\rvec{Y}} \Esp{\rveci{x}{j}|\rvec{Y}}\right)^2}.
\end{multline}
By simple inspection of the entries of the Hessian matrix above, the
result in \req{eq:HessEnt} can be found.

\section{Gradient and Hessian of $\Ent(\rvec{Z} = \rvec{X} +
\noiseR^{1/2}\rvec{W})$} \label{sec:Bakry}

The elements of the gradient of $\Ent(\rvec{Z}= \rvec{X} +
\noiseR^{1/2}\rvec{W})$ with respect to the diagonal elements of
$\noiseR$ can be found thanks to the complex multivariate de
Bruijn's identity found in \cite[Th. 4]{palomar:06} adapted to the
real case
\begin{gather} \label{eq:gradNw} 
\frac{\partial \Ent(\rvec{X} + \noiseR^{1/2}\rvec{W})}{\partial
\noiser{i}} = \frac{1}{2} \Esp{\left(\frac{\partial \log
\pdf{\rvec{z}}{\rvec{z}}}{\partial z_i} \right)^2}.
\end{gather}

The elements of the Hessian matrix can be found quite directly from
the expressions found in \cite[Eq. (50)]{costa:85} and in Villani's
Lemma in \cite{villani:00} for the single dimensional second
derivative $\d^2 \Ent(\rvec{X} + \sqrt{t}\rvec{W})/\d t^2$ (see
\cite{payaro:08} for further details on the specific generalization
to the multidimensional case):
\begin{gather} \label{eq:HessNw} 
\frac{\partial^2 \Ent(\rvec{X} + \noiseR^{1/2}\rvec{W})}{\partial
\noiser{i} \partial \noiser{j}} = -\frac{1}{2}
\Esp{\left(\frac{\partial^2 \log
\pdf{\rvec{z}}{\rvec{z}}}{\partial z_i \partial
z_j}\right)^2}.
\end{gather}

To further elaborate the expressions in \req{eq:gradNw} and
\req{eq:HessNw} we see that we need to compute the gradient and
Hessian of the function $\log \pdf{\rvec{z}}{\vec{z}}$. The
expression for the gradient has already been given in \cite[Eq.
(56)]{guo:05}, \cite[Eq. (105)]{palomar:06}
\begin{gather} \label{eq:dilogP}
\frac{\partial \log \pdf{\rvec{z}}{\rvecr{z}}}{\partial z_i} =
\frac{\Esp{\rveci{x}{i}|\rvec{z} = \rvecr{z}} -
\rvecri{z}{i}}{\noiser{i}}.
\end{gather}

The expression for the Hessian of $\log \pdf{\rvec{z}}{\vec{z}}$
requires slightly more elaboration and here we only give a sketch,
more details can be found in \cite{payaro:08}.

Differentiating \req{eq:dilogP} with respect to $z_j$ we obtain
\begin{multline} \label{eq:dijlogP}
\frac{\partial^2\log\pdf{\rvec{z}}{\rvecr{z}}}{\partial z_i \partial
z_j} = \frac{1}{\noiser{i} \noiser{j}} \left( \Esp{\rveci{x}{i}
\rveci{x}{j} | \rvec{z} = \rvecr{z}} \right.
\\ - \left. \Esp{\rveci{x}{i}|\rvec{z} = \rvecr{z}} \Esp{\rveci{x}{j}|\rvec{z} =
\rvecr{z}} \right) - \frac{\delta_{ij}}{\noiser{i}},
\end{multline}
where we have used that \cite{payaro:08}
\begin{multline}
\frac{\partial \Esp{\rveci{x}{i}|\rvec{z} = \rvecr{z}}}{\partial
z_j} = \frac{1}{\noiser{j}}\left(\Esp{\rveci{x}{i} \rveci{x}{j} |
\rvec{z} = \rvecr{z}} \right. \\ \left. - \Esp{\rveci{x}{i}|\rvec{z}
= \rvecr{z}} \Esp{\rveci{x}{j}|\rvec{z} = \rvecr{z}} \right).
\nonumber
\end{multline}

Plugging \req{eq:dilogP} into \req{eq:gradNw} and operating
according to the derivation in \cite[Eq. (106)]{palomar:06} we obtain
\begin{gather} \nonumber
\frac{\partial \Ent(\rvec{X} + \noiseR^{1/2}\rvec{W})}{\partial
\noiser{i}} = \frac{1}{2\noiser{i}^2}\left(\noiser{i} -
\Esp{(\rveci{x}{i} - \Esp{\rveci{x}{i}|\rvec{Z}})^2}\right).
\end{gather}

Similarly, plugging \req{eq:dijlogP} into \req{eq:HessNw} we obtain
the desired expression for the Hessian as
\begin{multline} \nonumber
\frac{\partial^2 \Ent(\rvec{X} + \noiseR^{1/2}\rvec{W})}{\partial
\noiser{i} \partial \noiser{j}} \\ = \! -\frac{1}{2} \EspOp \! \left\{\!\!\left(
\frac{\Esp{\rveci{x}{i}\rveci{x}{j} | \rvec{Z}} -
\Esp{\rveci{x}{i}|\rvec{Z}}
\Esp{\rveci{x}{j}|\rvec{Z}}}{\noiser{i}\noiser{j}} -
\frac{\delta_{ij}}{\noiser{i}} \right)^{\!\!2}\!\right\}, \nonumber
\end{multline}
which can be expanded as \vspace{-.12cm}
\begin{multline} \nonumber
\frac{\partial^2 \Ent(\rvec{X} + \noiseR^{1/2}\rvec{W})}{\partial
\noiser{i} \partial \noiser{j}} \\ =  -\frac{1}{2} \Esp{
\frac{\left(\Esp{\rveci{x}{i}\rveci{x}{j} | \rvec{Z}} -
\Esp{\rveci{x}{i}|\rvec{Z}}
\Esp{\rveci{x}{j}|\rvec{Z}}\right)^2}{\noiser{i}^2\noiser{j}^2}}
\\ - \frac{\delta_{ij}}{2\noiser{i}^2} +
\Esp{\frac{\Esp{\rveci{x}{i}^2 | \rvec{Z}} -
(\Esp{\rveci{x}{i}|\rvec{Z}})^2 }{\noiser{i}^3}}\delta_{ij}.
\nonumber
\end{multline}

\bibliographystyle{IEEEtran}
\bibliography{references_precoderdesign_arxiv}

\end{document}